\pdfoutput=1
\documentclass{article}

\title{Portfolio Optimization in the Stochastic Portfolio Theory Framework}
\author{Vassilios Papathanakos\footnote{Email: VPapathanakos@intechjanus.com}}
\usepackage{amsmath, amssymb, amsthm, bm, graphicx}
\newcommand\al\alpha
\newcommand\ga\gamma
\newcommand\de\delta
\newcommand\zi\zeta
\newcommand\Tht\Theta
\newcommand\tht\theta
\newcommand\La\Lambda
\newcommand\la\lambda
\newcommand\si\sigma
\newcommand\bb{\mathbf{b}}
\newcommand\eb{\mathbf{e}}
\newcommand\gb{\mathbf{g}}
\newcommand\pb{\mathbf{p}}
\newcommand\xb{\mathbf{x}}
\bmdefine\zerob{\mathbf{0}}
\bmdefine\alb{\mathbf{\al}}
\bmdefine\gab{\mathbf{\ga}}
\bmdefine\zib{\mathbf{\zi}}
\bmdefine\etab{\mathbf{\eta}}
\bmdefine\mub{\mathbf{\mu}}
\bmdefine\nub{\mathbf{\nu}}
\bmdefine\pib{\mathbf{\pi}}
\bmdefine\xib{\mathbf{\xi}}
\bmdefine\sib{\mathbf{\si}}
\bmdefine\taub{\mathbf{\tau}}
\newcommand\brac[1]{\left[ #1 \right]}

\newcommand\per[1]{\left( #1 \right)}
\newcommand\II{\mathbb{I}}
\DeclareMathOperator\diag{diag}
\newcommand\In{^{-1}}
\newcommand\Tr{^\intercal}
\newtheorem{defn}{Definition}
\newtheorem{lem}{Lemma}
\newtheorem{prop}{Proposition}
\newcommand\Req[1]{Equation~\eqref{eq:#1}}
\newcommand\Rfig[1]{Figure~\ref{fig:#1}}

\newcommand\pt\partial
\newcommand\p[2]{\frac{\pt #1}{\pt #2}}

\begin{document}
	\maketitle

	\tableofcontents

    \section{Introduction}

    Portfolio optimization is a fundamental concept in investing, but it presents many technical challenges. The two most important issues are the imperfect nature of the estimates of the market characteristics, and the ambiguity of the optimization objective.

    The market characteristics cannot be estimated with high accuracy because markets change at a pace that does not generally not permit the accumulation of enough relevant data to achieve statistical and non-singular convergence, e.g.,~in the estimation of the covariance matrix. One way forward is to exploit statistical relations in the estimate to identify combinations of statistics that are more resistant to uncertainty than the underlying data. Another approach is to assume a market model that uses powerful structural hypotheses to make up for the lack of data. The success of this approach depends crucially on the robustness and parsimony of these hypotheses. Proprietary statistical methodologies are often necessary to achieve these goals, especially when designing investment processes for institutional investors, who place a premium on long-term stability.

    Furthermore, even when focusing on a narrow investor segment, the optimization objective is typically fluid and vague. Even for a fixed investor, the relative desirability between outperformance potential and risk protection varies dramatically between market environments and cannot be easily quantified. Partly due to this, it is not practical to set up an optimization for a very long-term horizon (e.g.,~decadal spans of time). On the other hand, it is also not advisable to optimize exclusively on a very short time scale (e.g.,~hourly or daily), without a clear plan of how to avoid unnecessary turnover and other risks when joining the consecutive optimized portfolios together.

    In this write-up, I discuss some theoretical results with a view to motivate some practical choices in portfolio optimization. Even though the setting is not completely general (for example, the covariance matrix is assumed to be non-singular), I attempt to highlight the features that have practical relevance. In particular, the first two sections contain mathematical results that apply to portfolio optimization in the Stochastic Portfolio Theory (SPT) framework. This setting is flexible enough to describe most realistic assets, and it has been successfully employed for managing equity portfolios since 1987. The last section contains a discussion of some of the implications of these theoretical results for portfolio optimization in practice.

	\section{Fixed universe}

    In this section, I assume that the model parameters (i.e.,~the drifts and volatilities) are fixed in time. This allows for an explicit, time-independent solution of optimization problems.

	\subsection{Model and basic definitions}

	In following with the standard formulation of Stochastic Portfolio Theory \cite{F, Fa, FK}, consider an equity market composed of $n$ securities whose market capitalizations evolve as non-negative It\=o processes via
	\begin{equation}
		d\ln V_i(t)=\ga_i(t)\,dt+\sum_{l=1}^d\xi_{il}(t)\,dW_l(t)\,,
	\end{equation}
	where $n\ge2$, $d\ge n$, the $W$'s are independent and standard Brownian motions, and the $\ga$'s and $\xi$'s are measurable, adapted to the filtration generated by the $W$'s, and well-behaved\footnote{Essentially, the assumption is that they do not explode in a non-integrable sense within finite time, and that they do not grow too fast as time tends to infinity; for more details, cf.~\cite{Fa}.}. The $\xi$'s are non-degenerate.

	\begin{defn}
		The covariance matrix of the $V$'s is $\sib$, with elements
		\begin{equation}
			\si_{ij}(t)\equiv\sum_{k=1}^d\xi_{ik}(t)\,\xi_{jk}(t)\,;
		\end{equation}
		since $\sib$ is non-degenerate, it has the inverse $\sib^{-1}$. Also, $\eb$ denotes the column vector of ones, $\zerob$ the column vector of zeroes, and
		\begin{equation}\label{eq:Vol}
			s\equiv\frac{1}{\sqrt{\eb\Tr\cdot\sib\In\cdot\eb}}\,.
		\end{equation}
		Moreover, $\alb$ denotes the column vector of the individual stock arithmetic returns:
		\begin{equation}
			\alb\equiv\gab+\frac{1}{2}\,\diag\sib\,,
		\end{equation}
		where $\gab$ is the column vector of the individual stock growth rates, and $\diag\sib$ is the column vector of the diagonal elements of the covariance matrix $\sib$. Finally, define $a$ and $S$ through
		\begin{equation}\label{eq:Alp}
			a\equiv s^2\ \eb\Tr\cdot\sib\In\cdot\alb\,,
		\end{equation}
		and
		\begin{equation}\label{eq:VolA}
			S\equiv\sqrt{\alb\Tr\cdot\sib\In\cdot\alb-\frac{a^2}{s^2}+s^2}\,.
		\end{equation}
	\end{defn}

	\begin{lem}
		The argument of the square root in $S$ is positive; moreover,
		\begin{equation}
			S\ge s\,.
		\end{equation}
	\end{lem}

	\begin{proof}
		The Cauchy-Schwarz inequality implies that
		\begin{equation}
			\per{\alb\Tr\cdot\sib\In\cdot\alb}\per{\eb\Tr\cdot\sib\In\cdot\eb}\ge\per{\eb\Tr\cdot\sib\In\cdot\alb}^2\,,
		\end{equation}
		which implies
		\begin{equation}
			\alb\Tr\cdot\sib\In\cdot\alb-\frac{a^2}{s^2}\ge0\,.\qedhere
		\end{equation}
	\end{proof}

	I consider portfolios of securities that are expressed through their weights, as opposed to shares. Each portfolio $\pib$ is fully funded,
	\begin{equation}
		\sum_{i=1}^n\pi_i(t)=1\,,
	\end{equation}
	but I do not insist that the weights are all non-negative. The logarithmic return of the portfolio satisfies
	\begin{equation}
		d\ln V_\pib(t)=\sum_{i=1}^n\pi_i(t)\,d\ln V_i(t)+\ga_\pib^*(t)\,dt\,,
	\end{equation}
	where the excess-growth rate $\ga_\pi^*$ is given by
	\begin{eqnarray}
		\ga_\pib^*(t)&=&\frac{1}{2}\per{\sum_{i=1}^n\si_{ii}(t)\,\pi_i(t)-\sum_{i,j=1}^n\si_{ij}(t)\,\pi_i(t)\,\pi_j(t)}=\\
		&=&\frac{1}{2}\per{\diag\sib\Tr\cdot\pib-\pib\Tr\cdot\sib\cdot\pib}\,.
	\end{eqnarray}
	All results below are understood under the condition that $V_\pib$ remains positive up to the time $t$ under consideration.

	\subsection{Extremal portfolios}

	Consider the following two extremal portfolios:
	\begin{itemize}
		\item $\nub^{(0)}$: the minimum-volatility portfolio,
		\item $\nub^{(1)}$: the maximum-growth portfolio.
	\end{itemize}

	Their existence, uniqueness, and basic properties are the subject of the two propositions below.

	\begin{prop}
		The minimum-volatility portfolio $\nub^{(0)}$ exists and is unique. It is given by
		\begin{equation}
			\nub^{(0)}=s^2\, \sib\In\cdot\eb\,,
		\end{equation}
 	   has volatility
		\begin{equation}
			\si_{\nub^{(0)}}=s
		\end{equation}
		and growth rate
		\begin{equation}
			\gab_{\nub^{(0)}}=a-\frac{1}{2}\,s^2\,,
		\end{equation}
	\end{prop}

	\begin{proof}
		$\nub^{(0)}$ is the minimizer of the $\sib$-norm, so it is unique. The explicit expression for $\sib^{(0)}$ is derived using Lagrange multipliers. The expressions for the volatility and growth rate of $\nub^{(0)}$ are a direct consequence of the definitions in \Req{Vol} and \Req{Alp}, which they motivated.
	\end{proof}

	\begin{prop}\label{pro:nuba}
		The maximum-growth portfolio $\nub^{(1)}$ exists and is unique. It is given by
		\begin{equation}
			\nub^{(1)}=\sib\In\cdot\brac{\alb+\per{s^2-a}\eb}=\nub^{(0)}+\sib\In\cdot\per{\alb-a\,\eb}\,.
		\end{equation}
	    This portfolio has volatility
	   	\begin{equation}
	   		\si_{\nub^{(1)}}=S\,,
	   	\end{equation}
	   	and growth rate
	   	\begin{equation}
	   		\gab_{\nub^{(1)}}=a+\frac{1}{2}\,S^2-s^2=\ga_{\nub^{(0)}}+\frac{1}{2}\,\per{S^2-s^2}\,.
	   	\end{equation}
	\end{prop}

	\begin{proof}
		The uniqueness of $\nub^{(1)}$ follows by contradiction: suppose that there was a second portfolio, $\xib$, that also satisfies
		\begin{equation}
			\ga_\xib=\ga_{\nub^{(1)}}.
		\end{equation}
		Then, consider the portfolio $\pib^{(\la)}$
		\begin{equation}
			\pib^{(\la)}=\la\,\nub^{(1)}+\per{1-\la}\,\xib\,;
		\end{equation}
		this portfolio has growth rate
		\begin{equation}
			\ga_{\pib^{(\la)}}=\ga_{\nub^{(1)}}+\frac{\la\per{1-\la}}{2}\per{\nub^{(1)}-\xib}\Tr\cdot\sib\per{\nub^{(1)}-\xib}\,,
		\end{equation}
		which is strictly greater than $\ga_{\nub^{(1)}}$ for any $\la\in\per{0,1}$, unless $\xib\equiv\nub^{(1)}$, due to the non-degeneracy of $\sib$.

		The explicit expression for $\sib^{(1)}$ is derived using Lagrange multipliers. Finally, the expressions for the volatility and growth rate of $\nub^{(1)}$ follow directly from the definitions in \Req{VolA}, which they motivated.
	\end{proof}

	\begin{lem}
		The covariance of $\nub^{(0)}$ and $\nub^{(1)}$ equals the variance of $\nub^{(0)}$.
	\end{lem}

	This observation, which results from a straightforward calculation, can be interpreted as saying that the maximization of the growth rate by $\nub^{(1)}$ relies in maximizing the portfolio exposure to idiosyncratic sources of risk, which are orthogonal to the non-diversifiable equity risk (which is represented by the minimum-variance portfolio).

	\subsection{Efficient frontier}

	\begin{defn}
		Let $\nub^{(p)}$ denote the one-parametric portfolio family that interpolates $\nub^{(0)}$ and $\nub^{(1)}$:
		\begin{equation}
			\nub^{(p)}=\per{1-p}\,\nub^{(0)}+p\,\nub^{(1)}\,.
		\end{equation}
	\end{defn}

	\begin{prop}
		For every portfolio volatility $\si\in\brac{\si_{\nub^{(0)}},\si_{\nub^{(1)}}}$, there is a unique portfolio that maximizes the growth rate, and it is given by $\nub^{(p)}$ for the value $p\in\brac{0,1}$ for which $\si_{\nub^{(p)}}=\si$.
	\end{prop}

	\begin{proof}

		One way to show that $\nub^{(p)}$ is the unique maximizer is to use the explicit formulas for $\nub^{(0)}$ and $\nub^{(1)}$ in order to compute that the interpolated portfolio equals
		\begin{equation}\label{eq:Exp}
			\nub^{(p)}=\sib\In\cdot\brac{p\,\alb+\per{s^2-a\,p}\eb}=\nub^{(0)}+p\,\sib\In\cdot\brac{\alb-a\,\eb}\,,
		\end{equation}
		which implies its volatility equals
		\begin{equation}
			\si_{\nub^{(p)}}=\sqrt{\per{1-p^2}\,s^2+p^2\,S^2}
		\end{equation}
		and its growth rate
		\begin{equation}
			\ga_{\nub^{(p)}}=a+\per{p-\frac{p^2}{2}}\,S^2-\per{p+\frac{1-p^2}{2}}\,s^2\,.
		\end{equation}
		Using Laplace multipliers and the explicit expression for $\nub^{(p)}$ implies that it is the portfolio maximizing the growth rate for the given volatility.

		Another way to show the same result follows from parametrizing the optimal portfolio, $\etab^{(p)}$, via the difference portfolio, $\xib^{(p)}$,
		\begin{equation}
			\xib^{(p)}\equiv\etab^{(p)}-\nub^{(p)}\,,
		\end{equation}
		whence it follows that $\xib^{(p)}$ must satisfy the following two conditions:
		\begin{enumerate}
			\item the normalization of the portfolio weights implies
			\begin{equation}\label{eq:Fst}
				\eb\Tr\cdot\xib^{(p)}=0\,,
			\end{equation}
			\item the equality of the portfolio variances for $\nu^{(p)}$ and $\eta^{(p)}$ implies
			\begin{equation}\label{eq:Sec}
				\per{2\,\nub^{(p)}+\xib^{(p)}}\Tr\cdot\sib\cdot\xib^{(p)}=0\,.
			\end{equation}
		\end{enumerate}
		Substituting the explicit expression for $\nub^{(p)}$ from \Req{Exp} in \Req{Sec}, and using \Req{Fst} implies that the second condition is equivalent with
		\begin{equation}
			2\,p\,\alb\Tr\cdot\xib^{(p)}+\xib^{(p)}{}\Tr\cdot\sib\cdot\xib^{(p)}=0\,.
		\end{equation}

		Since the difference between the growth rates for $\etab^{(p)}$ and $\nub^{(p)}$ equals
		\begin{equation}
			\ga_{\etab^{(p)}}-\ga_{\nub^{(p)}}=\alb\Tr\cdot\xib^{(p)}=-\frac{1}{2\,p}\,\xib^{(p)}{}\Tr\cdot\sib\cdot\xib^{(p)}\leq0\,,
		\end{equation}
		the growth rate is maximized for $\xib^{(p)}\equiv\zerob$.
	\end{proof}

	\begin{lem}
		The efficient frontier has an infinite slope at $\nub^{(0)}$ and a zero slope at $\nub^{(1)}$.
	\end{lem}

	\begin{proof}
		This is a direct result of evaluating
		\begin{equation}
			\p{\si_{\nub^{(p)}}}{p}=p\,\frac{S^2-s^2}{\si_{\nub^{(p)}}}
		\end{equation}
		and
		\begin{equation}
			\p{\ga_{\nub^{(p)}}}{p}=\per{1-p}\per{S^2-s^2}
		\end{equation}
		at $p=0$ and $p=1$.
	\end{proof}

    \subsection{Risk-adjusted return}

	It is worth introducing the SPT-analog of the Sharpe ratio \cite{S}: assume that there is a benchmark investment, the risk-free asset, which delivers a return rate $b$ at zero volatility (in this case, the rate of return and the growth rate are equal).

    Then, one way to quantify the investment efficiency of a portfolio is via the ratio of the relative growth rate divided by the relative volatility:
	\begin{equation}\label{eq:Tht}
		\tht_\pib\equiv\frac{\ga_\pib-b}{\si_\pib}\,.
	\end{equation}
	In the case of the efficient frontier, this ratio equals
	\begin{equation}
		\tht_p=\frac{a-b+\per{p-\frac{p^2}{2}}\,S^2-\per{p+\frac{1-p^2}{2}}\,s^2}{\sqrt{\per{1-p^2}\,s^2+p^2\,S^2}}\,;
	\end{equation}
	its range of values are
	\begin{equation}
		\tht_0=\frac{a-b}{s}-\frac{s}{2}\le\tht_p\le\tht_1=\frac{a-b}{S}+\frac{S}{2}-\frac{s^2}{S}\,.
	\end{equation}

	\subsection{Below the efficient frontier}

	Consider two stocks, taken without loss of generality to have index 1 and 2 respectively, and construct the portfolio $\pib$ by interpolating between them with weights $x$ and $1-x$ respectively.

	Its growth rate equals the concave function
	\begin{equation}
		\ga_\pib=x\,\ga_1+\per{1-x}\ga_2+\frac{x\per{1-x}}{2}\per{\si_{11}+\si_{22}-2\,\si_{12}}\,,
	\end{equation}
	which has a maximum equal to
	\begin{equation}
		\ga_\pib^*=\frac{\ga_1+\ga_2}{2}+\frac{\per{\ga_1-\ga_2}^2}{2\per{\si_{11}+\si_{12}-2\,\si_{12}}}+\frac{\si_{11}+\si_{12}-2\,\si_{12}}{8}\,,
	\end{equation}
	at
	\begin{equation}
		x^*=\frac{1}{2}+\frac{\ga_1-\ga_2}{\si_{11}+\si_{12}-2\,\si_{12}}\,.
	\end{equation}

	Similarly, the variance of this two-stock portfolio equals the convex function
	\begin{equation}
		\si_\pib^2=x^2\,\si_{11}+\per{1-x}^2\,\si_{22}+2\,x\per{1-x}\,\si_{12}\,,
	\end{equation}
	which has a minimum equal to
	\begin{equation}
		\si_\pib^*=\frac{\si_{11}\per{\si_{22}-\si_{12}}^2+\si_{22}\per{\si_{11}-\si_{12}}^2+2\,\si_{12}\per{\si_{11}-\si_{12}}\per{\si_{22}-\si_{12}}}{\per{\si_{11}+\si_{22}-2\,\si_{12}}^2}
	\end{equation}
	at
	\begin{equation}
		x^*=\frac{\si_{22}-\si_{12}}{\si_{11}+\si_{22}-2\,\si_{12}}\,.
	\end{equation}

	These computations allow us to trace a parametric curve for the two-stock portfolio as the relative weight of the two stocks is varied as shown in \Rfig{TP}.

	\begin{figure}[!htbp]
		\centering
		\includegraphics[width=0.5\textwidth]{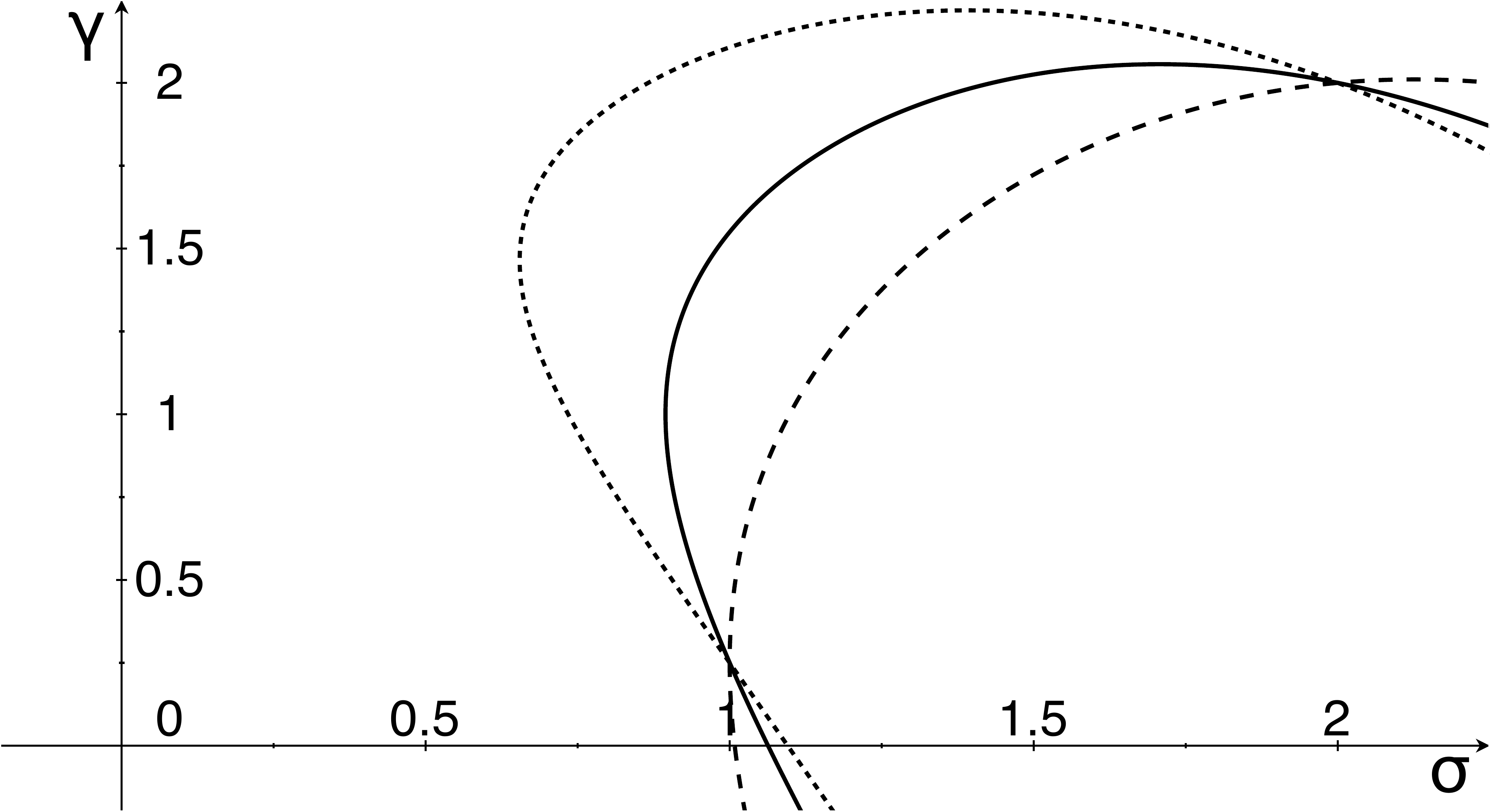}
		\caption{Behavior of the portfolio volatility and growth rate for a two-stock portfolio for $\si_{12}=0$ (continuous line), $\si_{12}=-1$ (dotted line), and $\si_{12}=1$ (dashed line). I assume that $\ga_1=\frac{1}{4}$, $\ga_2=2$, $\si_{11}=1$ and $\si_{22}=4$, and vary $x$ in the interval $\brac{-1,2}$. All lines intersect at the two points corresponding to the underlying stocks.}
		\label{fig:TP}
	\end{figure}

    It is clear that the worst-performing all-long portfolio is the one investing exclusively in the stock with the lowest growth rate. Similarly, the maximum volatility for the long-only case is realized by investing exclusively in the most volatile stock.

	The same argument can be applied recursively for more than two stocks, which implies that the minimum growth rate for a long-only portfolio is realized by investing exclusively in the worse performing stock. Similarly, the maximum volatility for a long-only portfolio is realized by investing exclusively in the most volatile stock. On the other hand, there is no lower bound for growth and upper bound for volatility concerning long-short portfolios.

	\subsection{Volatility-stabilized markets}

    When the model parameters are not fixed, the above results require adjustments. One of the simplest extensions is when the model parameters vary continuously, so that the portfolio optimization can be performed in a continuously-updated manner.

    As an example of such a scenario, consider the case of the volatility-stabilized market model \cite{FKa}, where
	\begin{equation}
		d\ln V_i(t)=\frac{dW_i(t)}{\sqrt{\mu_i(t)}}\,,
	\end{equation}
	with $\mu$'s being the market portfolio weights, namely
	\begin{equation}
		\mu_i(t)\equiv\frac{V_i(t)}{\sum_{j=1}^nV_j(t)}\,.
	\end{equation}

	The covariance matrix in this model equals
	\begin{equation}
		\si_{ij}=\frac{\de_{ij}}{\mu_i}\,,
	\end{equation}
	where $\de$ is the Kronecker delta. Its inverse matrix equals
	\begin{equation}
		\per{\si\In}_{ij}=\de_{ij}\,\mu_i\,,
	\end{equation}
	and the minimum-variance portfolio equals
	\begin{equation}
		\nub^{(0)}=\mub\,,
	\end{equation}
	where $\mub$ is the column vector of market weights; $\nub^{(0)}$ has the properties
	\begin{equation}
		\si_{\nub^{(0)}}=s=1\,,\qquad\ga_{\nub^{(0)}}=\frac{n-1}{2}\,.
	\end{equation}

	Regarding $\nu^{(1)}$, it is given by
	\begin{equation}
		\nub^{(1)}=\frac{\eb}{2}+\per{1-\frac{n}{2}}\mub\,,
	\end{equation}
	with properties
	\begin{equation}
		\si_{\nub^{(1)}}=S=\frac{1}{2}\,\sqrt{D_{-1}^{-1}+4-n^2}\,,\qquad\ga_{\nub^{(1)}}=\frac{4\,n+D_{-1}^{-1}-n^2-4}{8}\,,
	\end{equation}
	where
	\begin{equation}
		D_{-1}\equiv\frac{1}{\sum_{i=1}^n\frac{1}{\mu_i}}\,;
	\end{equation}
    note that in the expression $D_{-1}^{-1}$, the superscript denotes an exponent.

    \subsubsection{Entropy-weighted portfolio}

	As an aside, consider the long-only portfolio $\zib$ with weights
	\begin{equation}
		\zi_i(t)=\frac{\mu_i(t)\per{c-\ln\mu_i(t)}}{Z(\mub(t))}\,,\qquad Z(\xb)\equiv c-\sum_{i=1}^nx_i\,\ln x_i\,,
	\end{equation}
	because of its role in the theory of volatility-stabilized markets \cite{FKa}. Its volatility and growth rate are
	\begin{equation}
		\si_\zib^2=\frac{\sum_{i=1}^n\mu_i\,\ln\mu_i^2+2\,c\,Z(\mub)-c^2}{Z^2(\mub)}\,,\qquad
		\ga_\zib=\frac{n\,c-\sum_{i=1}^n\ln\mu_i}{2\,Z(\mub)}-\frac{\si_\zib^2}{2}\,.
	\end{equation}

    \subsubsection{$n=3$ case}
	The largest $n$ for which both $\nub^{(0)}$ and $\nub^{(1)}$ are all-long in all possible market configurations equals
	\begin{equation}
		n=3\,.
	\end{equation}
	In that case, the maximum-growth portfolio becomes simply the equal-weighted portfolio and the above formulas can be simplified to
	\begin{equation}
		\nub^{(0)}=\mub\,,\qquad\si_{\nub^{(0)}}=1\,,\qquad\ga_{\nub^{(0)}}=1\,,
	\end{equation}
	and
	\begin{equation}
		\nub^{(1)}=\frac{1}{2}\,\per{\eb-\mub}\,,\quad\si_{\nub^{(1)}}=\frac{1}{2}\,\sqrt{D_{-1}^{-1}-5}\,,\quad\ga_{\nub^{(1)}}=\frac{D_{-1}^{-1}-1}{8}\,.
	\end{equation}

	The maximum value of $D_{-1}$ is $\tfrac{1}{9}$ and occurs for the equal-weighted market portfolio. In this case, the minimum-volatility and the maximum-growth portfolio coincide, and the efficient frontier degenerates to a point. Also, the risk-adjusted return equals (cf.~\Req{Tht})
	\begin{equation}
		\tht=1-b\,.
	\end{equation}

	The efficient frontier is given by
	\begin{equation}
		\per{\si_p,\ga_p}=\per{1+\per{\frac{p^2}{2}-p}\frac{9-D_{-1}^{-1}}{4}}\,,
	\end{equation}
	where $p\in\brac{0,1}$, and it is realized by long-only portfolios.

	\section{Variable universe}

    In a market where the model parameters (i.e.,~drifts and volatilities) vary over time in a stochastic manner, it is not possible to perform a long-term optimization as in the previous section. In general, the best one can hope for is the situation where the parameters vary continuously, so that the solution to the optimization problem can be also continuously updated, as was the case in the previous subsection. This situation can be treated within the framework of dynamic control. The resulting analysis is very interesting from a mathematical point of view; however, incorporating it in practical investment processes in a robust manner presents great challenges.

    In spite of all that, in models where the market structure depends only on the ranks of various securities, it is possible to recover a version of the long-term optimization problem \emph{relative} to the performance of the market portfolio. This is the subject of this section.

	\subsection{Model and definitions}

    Consider a version of the Atlas model \cite{IPBK} described by
	\begin{equation}
		d\ln V_i(t)=g_{r_i(t)}\,dt+\sum_{l=1}^d\xi_{r_i(t)l}\,dW_l(t)\,,
	\end{equation}
	where $r_i(t)$ is the rank of stock $i$ (ordered in descending order by market capitalization; ties are resolved lexicographically).

	In this market model, stocks regularly exchange their growth rates and volatilities, but the rank-specific parameters are fixed. In order to perform a portfolio optimization, I rely on the theory of functionally-generated portfolios \cite{Fc, Fa}. The main result employed from this theory is that the performance of a portfolio $\pib$ that is fixed in rank terms,
	\begin{equation}
		\pi_i(t)=p_{r_i(t)}\,,
	\end{equation}
	relative to the market portfolio is given by
	\begin{equation}
		d\ln\frac{V_\pib(t)}{V_\mub(t)}=d\ln F_\pb(\mub_{(\cdot)}(t))+d\Tht(t)\,,
	\end{equation}
	where $\mub_{(\cdot)}(t)=\per{\mu_{(i)}(t)}_{i=1}^n$ is the vector of ranked market weights, $F$ is the generating functional,
	\begin{equation}
		F(\xb)=\prod_{i=1}^nx_i^{p_i}\,,
	\end{equation}
	and
	\begin{equation}\label{eq:Tht}
		d\Tht(t)=\ga_p^*\,dt+\frac{1}{2}\,\sum_{i=1}^{n-1}\per{p_{i+1}-p_i}\,d\La_i(t)\,,
	\end{equation}
    where
    \begin{equation}
      \ga_p^*=\frac{1}{2}\,\per{\sum_{i=1}^np_i\,\si_{(ii)}(t)-\sum_{i,j=1}^np_i\,p_j\,\si_{(ij)}(t)}\,;
    \end{equation}
    finally, $\La$'s are the local times at consecutive ranked market weights \cite{Fa, KS},
	\begin{equation}
		\La_i(t)=\La_{\ln\mu_{(i)}-\ln\mu_{(i+1)}}(t)\,.
	\end{equation}

	The stability of the Atlas model requires \cite{Fa, FK}
	\begin{equation}\label{eq:Stab}
		g_1<0\,,\qquad g_1+g_2<0\,,\quad\ldots\quad g_1+\ldots+g_{n-1}<0\,,\quad g_1+\ldots+g_n=0\,.
	\end{equation}
	Moreover, the local-times have a simple expression in terms of the growth rates \cite{Fa, FK}
	\begin{equation}\label{eq:loc}
		\lim_{T\to\infty}\frac{\La_i(T)}{T}=-2\,\sum_{j=1}^ig_j\,.
	\end{equation}

    \begin{prop}
      The long-term performance of a fixed-by-rank portfolio $\pib$ relative to the market is given by
      \begin{equation}
        \lim_{T\to\infty}\per{\frac{1}{T}\,\ln\frac{V_\pib(T)}{V_\mub(T)}}=\gb\Tr\cdot\pb+\ga_\pb^*\,.
      \end{equation}
    where $\gb$ is the column vector of the rank-based growth rates.
    \end{prop}

    \begin{proof}
      Substituting \Req{loc} to the second term of \Req{Tht}, and applying summation by parts results in
      \begin{equation}
        \lim_{T\to\infty}\per{\frac{1}{T}\,\frac{1}{2}\,\sum_{i=1}^{n-1}\per{p_{i+1}-p_i}\,d\La_i(T)}=\sum_{i=1}^{n-1}p_i\,g_i-p_n\,\sum_{i=1}^{n-1}g_i\,.
      \end{equation}
      Using the last of \Req{Stab}, and exploiting the stability of the market,
      \begin{equation}
        \lim_{T\to\infty}\per{\frac{1}{T}\,\ln\frac{F_\pb(\mub_{(\cdot)}(T))}{F_\pb(\mub_{(\cdot)}(0))}}=0\,,
      \end{equation}
      completes the proof.
    \end{proof}

    The above proposition implies that the results in the first section regarding minimum-variance and maximum-growth portfolios can be extended to the case where model parameters are rank-based. The only changes are that
    \begin{itemize}
    	\item[a)] they apply asymptotically, as opposed to any particular instant in time,
        \item[b)] the performance benchmark is the market portfolio, instead of cash.
    \end{itemize}

	\subsection{A simple Atlas model}

    As an application, consider the simplified Atlas model described by
	\begin{equation}
		d\ln V_i(t)=\per{-g+n\,g\,\II\brac{r_i(t)=n}}\,dt+\si\,dW_i(t)\,,
	\end{equation}
	where $g$ and $\si$ are positive constants.

	The rank-based covariance matrix equals
	\begin{equation}
		\si_{(ij)}=\si^2\,\de_{ij}\,,
	\end{equation}
	and its inverse is
	\begin{equation}
		\per{\sib\In}_{(ij)}=\frac{\de_{ij}}{\si^2}\,.
	\end{equation}

	The minimum-variance portfolio is
	\begin{equation}
		\nub^{(0)}=\frac{\eb}{n}\,,
	\end{equation}
	and has
	\begin{equation}
		\si_{\nub^{(0)}}=s=\frac{\si}{\sqrt{n}}\,,\qquad\ga_{\nub^{(0)}}=\frac{\si^2}{2}\per{1-\frac{1}{n}}\,;
	\end{equation}
    note that this is a constant portfolio in both name- and rank-based formulations, as opposed to only in the rank-based formulation (which is the generic case).

	The maximum-growth portfolio is
	\begin{equation}
		\nub^{(1)}=\per{\frac{1}{n}-\la}\eb+\la\,n\,\bb\,,
	\end{equation}
	where
	\begin{equation}
		\la\equiv\frac{g}{\si^2}\,,
	\end{equation}
	and $\bb$ is the column vector of which the only non-vanishing element, for the stock at the bottom rank, equals one:
	\begin{equation}
		b_i=\de_{in}\,.
	\end{equation}
	The properties of $\nub^{(1)}$ are
	\begin{equation}
		\si_{\nub^{(1)}}=S=\frac{\si}{\sqrt{n}}\,\sqrt{1+n^2\per{n-1}\la^2}\,,\ \ga_{\nub^{(1)}}=\frac{\si^2}{2}\per{1-\frac{1}{n}+n(n-1)\la^2}\,.
	\end{equation}

	Both $\nub^{(0)}$ and $\nub^{(1)}$ are all-long for all possible market configurations, if
	\begin{equation}
		g\le\frac{\si^2}{n}\,,
	\end{equation}
	or, equivalently,
	\begin{equation}
		\la\le\frac{1}{n}\,.
	\end{equation}
	However, $\la=\frac{1}{n}$ results in $\nub^{(1)}$ investing exclusively in the bottom stock; this is the ultimate small-cap portfolio.

    \subsubsection{Diversity-weighted portfolio}

	As an aside, consider the long-only portfolio $\zib$ with weights
	\begin{equation}
		\zi_i(t)=\frac{\mu_i^p(t)}{D_p(\mub(t))}\,,\qquad D_p(\xb)\equiv\per{\sum_{i=1}^nx_i^p(t)}^{\frac{1}{p}}\,,
	\end{equation}
	this diversity-weighted portfolio helps explore the impact of size to performance relative to a stable market \cite{FKK}; note that $p=1$ corresponds to the market portfolio, while $p=0$ corresponds to the equal-weighted portfolio.

	Its volatility and growth rate are
	\begin{equation}
		\si_{\zib}=\si\,\frac{\sum_{i=1}^n\mu_i^{2\,p}}{\per{\sum_{i=1}^n\mu_i^p}^2}=\si\,\frac{D_{2\,p}^{2\,p}\per{\mub}}{D_p^{2\,p}\per{\mub}}
	\end{equation}
	and growth rate
	\begin{equation}
		\ga_{\zib}=\frac{\si^2}{2}\brac{\frac{1}{2}\per{1-\frac{D_{2\,p}^{2\,p}\per{\mub}}{D_p^{2\,p}\per{\mub}}}-\la+n\,\la\,\frac{\mu_{(n)}^p}{D_p^p\per{\mub}}}\,.
	\end{equation}

	\section{Discussion}

    A crucial observation about asset classes is that they exhibit performance regimes: the drifts and volatilities of the investable universe change dramatically between different periods. This causes serious challenges in properly implementing portfolio optimization, including:
    \begin{itemize}
    	\item balancing accuracy and timeliness in estimating the changes in the model parameters;
        \item reconciling the differing time scales over which investors evaluate performance, markets evolve, and portfolio is implemented;
        \item avoiding frictional costs due to portfolio changes that may impact both the short- and long-term performance of an investment strategy.
    \end{itemize}

    There is no single approach that addresses all of these issues. Still, the adoption of some broadly diversified portfolio as the return benchmark for each optimization appears to be a vital component for ensuring consistent performance over the long term.

    So far, there have been three main lines of evidence supporting this view, all of which are empirical in nature:
    \begin{enumerate}
    	\item Using a broad benchmark establishes a context for what returns are reasonably achievable by investing in a class, while avoiding extreme concentration or overreliance on the estimates of the model parameters.
        \item In addition, the accuracy of various estimates is increased when working with relative quantities (computed relative to the broad benchmark), as opposed to absolute quantities.
        \item Finally, linking the optimization solution for multiple consecutive periods is facilitated by employing a relative objective function (e.g., excess return relative to the broad benchmark).
    \end{enumerate}

    The results above, especially those contained in the second section, furnish novel, theoretical support for this view, at least in the case of equity markets where Atlas models have been shown to be reasonable approximations.

    \subsubsection*{Acknowledgement}

    I would like to thank Bob Fernholz, for composing a problem set that inspired these notes. I also had useful discussions with Adrian Banner, Ioannis Karatzas, and Phillip Whitman.


\end{document}